\pdfoutput=1
\newif\ifFull
\Fullfalse
\documentclass[11pt, letterpaper]{article}
\advance \topmargin by -\headheight
\advance \topmargin by -\headsep
\evensidemargin \oddsidemargin
\usepackage{graphicx}
\usepackage{amsfonts}
\usepackage{url}
\usepackage{times}
\usepackage{epstopdf}
\usepackage{subfigure}
\usepackage{boxedminipage}
\usepackage{amsmath}
\usepackage{algorithmic}
\usepackage[lined,boxed]{algorithm2e}
\usepackage[utf8]{inputenc}
\usepackage{hyperref}
\usepackage[xindy]{glossaries}

\newcommand{\e}{{\rm e}}
\newcommand{\E}{{\rm E}}

\newtheorem{cond}{Condition}
\newtheorem{theorem}{Theorem}

\newenvironment{proof}{\noindent{\bf Proof:}}{\hspace*{\fill}\rule{6pt}{6pt}\bigskip}

\author{Michael Mitzenmacher\thanks{School of Engineering and Applied Sciences, Harvard University.
This work was supported by NSF grants CCF-0915922, IIS-0964473,CNS-1228598, and CCF-1320231.  Part of this work was done while visiting Microsoft Research New England.}}
\title{A New Approach to Analyzing Robin Hood Hashing}

\date{}

\begin{document}

\maketitle

\vspace{-0.3in}

\begin{abstract}
Robin Hood hashing is a variation on open addressing hashing designed
to reduce the maximum search time as well as the variance in the
search time for elements in the hash table.  While the case of
insertions only using Robin Hood hashing is well understood, the
behavior with deletions has remained open.  Here we show that Robin
Hood hashing can be analyzed under the framework of finite-level
finite-dimensional jump Markov chains.  This framework allows us to
re-derive some past results for the insertion-only case with some new
insight, as well as provide a new analysis for a standard deletion
model, where we alternate between deleting a random old key and
inserting a new one.  
\end{abstract}
\thispagestyle{empty}

\newpage

\setcounter{page}{1}

\vspace{-0.35in}
\section{Introduction}
\vspace{-0.05in}
Robin Hood hashing is a variation on open addressing hashing designed
to reduce the maximum search time as well as the variance in the
search time for elements in the hash table.  Here we are interested in
the setting where the probe sequences are random.  We briefly describe
the setup, starting with a setting with insertions only.  We have a hash table with $n$ cells, and $m = \lceil
\alpha n \rceil$ keys to place in the table.  We refer to $\alpha$ as
the load of the table; generally, we assume $\alpha n$ is an integer
henceforth. Each key $K_i$ has an associated infinite probe sequence
$K_{ij}$, with $j \geq 1$, where the $K_{ij}$ are uniformly
distributed over $[0,n-1]$.  Equivalently, the $K_{ij}$ are determined
by a random hash function $h$, where for a keyspace $K$ the hash
function has the form $h : K \times \mathbb{N} \rightarrow [0,n-1]$.\footnote{Alternatively,
we could have each probe sequence be a random permutation of $[0,n-1]$ for each key;
for our purposes, the two models are essentially equivalent, and we use
the random hash function model as it is easier to work with.}
Each key will be placed according to a position in its probe sequence.
If the $i$th element is placed in cell $K_{ij}$, and there is no $j' <
j$ such that $K_{ij'} = K_{ij}$, we shall say that the {\em age} of
the key is $j$.  If we use the standard search process for an key, by
which we mean sequentially examining cells according to the probe
sequence, the age of a key in the table corresponds to the number of
cells that must be searched to find it.  We assume that we keep track
of the age of the oldest key in the table.  In the standard search
process, one determines that a key not in the table is not present by
sequentially examining cells according to the probe sequence until
either an empty cell is found, or one has found that the key being
searched for must be older than oldest key in the table.  An empty cell
provides a witness that the key is not in the table.  We refer to a
search for a key not in the system as an unsuccessful search.

For the insertion of keys in the table, we may think of the keys as being
{\em placed} sequentially, using the probe sequence in the following
manner.  If $K_{i1}$ is empty when the $i$th key is inserted into the
table, the key is readily placed at cell $K_{i1}$.  Otherwise, there
is a collision, and a collision resolution strategy is required.  The
main point of Robin Hood hashing is that it {\em resolves collisions in
  favor of the key with the larger age}; the key with the smaller age
must continue sequentially through its probe sequence.  Notice that,
under Robin Hood hashing, a placed key will be displaced by the key
currently being placed if the placed key's age is smaller.  In this
case the placed key is moved from its current cell and becomes the
item to be placed, consequentially increasing its age.  Other standard
conflict resolution mechanisms are first come first served and last
come first served.  By favoring more aged keys, Robin Hood hashing
aims to reduce the maximum search time required.

Most of the results for Robin Hood hashing appear in the thesis of Celis
\cite{celisphd}, who provides a number of theoretical and empirical results.
(See also \cite{celis1985robin}.)  The following results are especially
worth mentioning.  First, when there are only
insertions, Celis analyzes the asymptotic behavior of Robin Hood
hashing (in the infinite limit setting) for loads $\alpha < 1$.  We
describe this result further in Section~\ref{sec:agedist}.  Second,
Celis shows that the total expected insertion cost in terms of the
number of probes evaluated by the standard insertion process -- or
equivalently the average age of keys in the table -- is the same for a
class of ``oblivious'' collision resolution strategies that do not
make use of knowledge about the future values in the probe sequences
and that include Robin Hood hashing (as well as first come first served
and last come first served).  Third, Devroye, Morin, and Viola have
shown that for $\alpha < 1$ the maximum search time for Robin Hood
hashing is (upper and lower) bounded by $\log_2 \log_2 n
\pm O(1)$ with probability $1-o(1)$, where the $O(1)$ terms depend on
$\alpha$ \cite{devroye2004worst}.  This double-logarithmic behavior
also occurs with quite different hashing schemes based on the power
of multiple choices \cite{ABKU,TwoSurvey}.  Finally, we note that Robin
Hood hashing has been also studied extensively in the setting of linear probing
schemes \cite{janson2005individual,viola1995analysis,viola2005exact}.

Here we provide a new approach for analyzing Robin Hood
hashing, based on a fluid limit analysis utilizing differential
equations.  An interesting aspect of our analysis is that it requires
using an additional level parameter, corresponding to a faster-moving
Markov process (tracking the age of current key being placed) beyond
the larger-scale Markov process (tracking the distribution of ages in
the table).  This type of analysis was previously used to study load
balancing schemes with
memory \cite{luczak2013averaging,mitzenmacher2002load}.  Our analysis
allows us to re-derive previous results for Robin Hood hashing, such as the asymptotic
behavior for loads $\alpha < 1$, while also
providing some additional novelty, such as concentration bounds for
finite $n$.  We also re-derive the high probability upper bound on the
maximum search time for Robin Hood
hashing of $\log_2 \log_2 n + O(1)$ of \cite{devroye2004worst}, with what
we suggest is a simpler and more intuitive proof.
More importantly, our approach is amenable to studying
Robin Hood hashing with deletions of random keys, an area that lacked a theoretical
framework for analysis previously.  We study the deletion scheme proposed 
by Celis in \cite{celisphd} under the setting of random deletions of keys and
new keys being inserted (maintaining a constant load $\alpha$), and 
suggest and analyze an alternative deletion scheme that is
simpler for practical implementations.

Before beginning, we remark that practical use is not our main
motivation for studying the Robin Hood hashing variant we examine here
(although we have seen some suggestions that it is still used on
occasion).  Robin Hood hashing can require substantially more
randomness than many other hashing schemes (such as cuckoo hashing or
other multiple-choice hashing schemes), and practical considerations
such as cache performance and prefetching suggest that one would
prefer to use the linear probing variant in almost all settings.  Our
motivation instead is in understanding this classical and
combinatorially simple-seeming hashing scheme, as well as in the
techniques that can be used to analyze it.  In particular, the 
double logarithmic bound on the search time requires some non-trivial additional
technical work beyond the the standard layered induction approach, and our analysis
of performance with deletions appears entirely new.   

In what follows, we provide background on the fluid limit approach we use here.
We then study Robin Hood hashing in the setting of insertions only under
this framework, and subsequently move on to examining how to analyze settings
with random deletions.  We note that our work includes extensive simulations that demonstrate
show the accuracy of our approach.  Due to space limitations, we put these in appendices.

\vspace{-0.15in}
\section{Limiting Framework}
\vspace{-0.1in}

For our limiting framework, we can work in the setting of finite-level
finite-dimensional jump Markov chains.  Here we roughly follow the
exposition of \cite{mitzenmacher2002load}; further development can be found in
\cite{ShwartzWeiss}.  Our discussion here is brief, 
and may be skipped by the uninterested reader willing 
to accept the more intuitive explanations that follow.  
We suspect
this methodology should be useful for studying other hashing variations.

In our setting, a chain with $D$ dimensions and $L$ levels will have
the state space ${\mathbb R}^D \times \{1,2,\ldots,L\}$.  The state can be
represented as a $D+L$-tuple in the natural way as follows:
$(\bar{x};m) = (x_1,\dots,x_{D};0,\dots,1,\dots,0)$, where a $1$ in
position $D+m, 1\leq m\leq L$, represents that the system is in level
$m$.  When in state $(\bar{x};m)$ the system can make $\zeta(m)$
possible different jumps.  Here we describe only unit jumps based
on unit vectors, which suffices for our main application, but more general
jumps are possible.  The process jumps to state $(\bar{x}+\bar{e}_i(m);g_i(m)) =
(\bar{x}+\bar{e}_i(m); 0,\dots,1,\dots,0)$
with rate $\nu_i(\bar{x};m)$, for $1\leq i\leq \zeta(m)$, where the 1 is
in position $D+g_i(m)$.
Here $\bar{e}_i(m)$ is a unit vector in one of the $D$ dimensions, and
$g_i(m)$ is the (new) level associated with the $i$th of the $\zeta(m)$ possible jumps;
note that $g_i(m)$ might itself be $m$,
so that the level may not change.  
The high-level idea is that here we have an underlying
finite-dimensional jump Markov process, but we also have an
additional associated ``level'' process that may drive the transition
rates of the primary jump Markov process.

The generator $A$ of this Markov process, which operates on real valued
functions $f: {\mathbb R}^{D+L}\to{\mathbb R}$, is defined as:
\vspace{-0.05in}
\begin{eqnarray}
 Af(\bar{x};m) & = & \sum_{i=1}^{\zeta(m)} \nu_i(\bar{x};m)
 [f(\bar{x}+\bar{e}_i(m);g_i(m))-f(\bar{x};m)]  \label{gen1}
\end{eqnarray}
We now consider a scaled version of this process, with scaling parameter $n$,
where the rate of each transition is scaled up by a factor of $n$ and
the jump magnitude is scaled down by a factor of $n$. The state of this scaled
system will be represented by $(\bar{s}_n;m) = (s_1,\dots,s_{D};0,\dots,\frac{1}{n},\dots,0)$
(with now a $1/n$ term in the position for level $m$).  
The associated jump vectors will be $(\frac{\bar{e}_i}{n};0,\dots,-\frac{1}{n},\dots,\frac{1}{n},\dots,0)$,
with corresponding rates are $n\nu_i(\bar{s}_n;m)$ for $1\leq i\leq\zeta(m)$.
(Note that in the case where the level does not change, the $-\frac{1}{n}$ and $\frac{1}{n}$ jumps
in the level should be interpreted as being in the same coordinate, so no change occurs.)
The generator for the scaled Markov process is:
\vspace{-0.05in}
\begin{eqnarray}
\label{gen2}
A_nf(\bar{s}_n;m) & = & \sum_{i=1}^{\zeta(m)} n\nu_i(\bar{s}_n;m)
 \{ f[\bar{s}_n+ \frac{\bar{e}_i(m)}{n};g_i(m)] - f(\bar{s}_n;m)\}. 
\end{eqnarray}
The following theorem (Theorem 8.15 from \cite{ShwartzWeiss}) describes the
evolution of the typical path of the scaled Markov process in the limit as
$n$ grows large.  The idea behind the theorem is that because the
finite-level Markov chain reaches equilibrium in some finite time, for large
enough $n$ the approximation that the finite-level Markov chain is in
equilibrium is sufficient to obtain Chernoff-like bounds.  

\begin{theorem}
\label{ldthm}
Under Conditions \ref{c1} and \ref{c2} below, for any given $T$ and constant $\epsilon > 0$, there 
exist positive constants $C_1,C_2(\epsilon)$ and $n_0$ such that for all initial
positions $\bar{s}^0 \in {\mathbb R}^D$, any initial level $ m\in \{0,1,\dots,L-1\}$, and 
any $n \geq n_0$,
\begin{eqnarray}
\Pr_{\bar{s}^0,m}\left( \sup_{0\leq t \leq T}\left|\bar{s}_n(t)-\bar{s}_{\infty}(t)\right|>\epsilon \right)
            & \leq & C_1 \exp(-nC_2(\epsilon)),
\end{eqnarray}
where $\bar{s}_{\infty}(t)$ satisfies the following:
\vspace{-0.05in}
\begin{eqnarray}
\frac{d}{dt}\bar{s}_{\infty}(t) & = & \sum_{l=0}^L \Pr(m(t)=l) \sum_{i=1}^{\zeta(l)}
  \nu_i(\bar{s}_{\infty};l)\bar{e}_i(l) \ ; \label{diffeq} \\
\bar{s}_{\infty}(0) & = & \bar{s}^0. \nonumber
   \end{eqnarray}
Here $\Pr(m(t)=l)$ is the equilibrium probability of the level-process being in level $l$
given the state $\bar{s}_{\infty}(t)$.
\end{theorem}

\vspace{-0.05in}
\begin{cond}\label{c1}
For any fixed value of $\bar{x} \in {\mathbb R}^D$, the Markov process evolving over 
the levels $\{1,\dots,L\}$ with transition rate $\nu_i(\bar{x};m)$ of going to level 
$g_i(m)$ from level $m$, is ergodic.
\end{cond}
\vspace{-0.15in}
\begin{cond}\label{c2}
The functions $\log \nu_i(\bar{x};y)$ are bounded and Lipschitz continuous
in $\bar{x}$ for every $y$ (where continuity is in all the $D$ coordinates).
\end{cond}
\vspace{-0.05in}

We note that a limitation of this approach is that it directly provides bounds
only on the finite-dimensional version of the process.  It is thus not
immediate that one can obtain rigorous bounds on the maximum search time for
Robin Hood of the form $\log_2 \log_2 n \pm O(1)$ as
in \cite{devroye2004worst} directly using this approach, as tracking
$\log_2 \log_2 n$ dimensions takes us outside the finite-dimensional realm. Instead, one
may use these results as useful intuition for guiding non-limiting
probabilistic arguments such as that derived
in \cite{devroye2004worst}, as we show here.  In return for this limitation, however, this approach
provides simple and general means for generating rich, accurate numerical results
that can aid in design and performance testing for real-world
implementations.

\vspace{-0.15in}
\section{Robin Hood Hashing with Insertions Only}
\vspace{-0.1in}
\subsection{Applying the Limiting Framework}
\vspace{-0.05in}

We first describe the standard process, which corresponds to the unscaled process
described above;  we generally use the term unscaled process where
the meaning is clear.
Each time step corresponds
to an attempted placement of a key, which can either be a new key, or
a key that was not successfully placed at the last time step, or a key
that was displaced by another key at the last time step.

To keep track of the state, we take advantage of the fact that keys
are placed randomly into cells.  Hence, for the state is suffices to
track the number of cells holding keys for each age; their actual
position does not matter.  Each time step corresponds to an attempt to
place a key.  Note that the number of time steps here is {\em not}
equal to the number of keys placed; placing a new key can take several
time steps with Robin Hood hashing, and as discussed during that
process the key being placed can take the place of another key which
then has to be placed.  Each such placement attempt represents a time
step.  As is often the case with hashing schemes, we find it more
useful to look at the tails of the loads rather than the loads
themselves.  Therefore, we let $x_i(t)$ be the fraction of cells with
a key with age at least $i$ after $t$ unscaled time steps.  For our
scaled version of the state, we let $s_i(t)$ be the fraction of cells
holding a key of age at least $i$ afer $tn$ key placements have been
tried, so that $x_i(nt) = s_i(t)$.  The level in our process will
correspond to the age of key currently being placed.  Fresh keys that
are newly being inserted have age 1.

We note that, as described, the process is infinite-dimensional, in 
that we can consider the values $s_i$ for all $i \geq 1$.  Indeed,
this is usually how we will think of the process, although as we show
later we can ``truncate'' the system at any finite value of $i$, which
can allow us to apply Theorem~\ref{ldthm}.

As a warm-up in understanding the scaled process, note that when the
load of the table is $z$, so that $zn$ cells contain a key, the number
of time steps to place an element is geometrically distributed with
mean $1/(1-z)$.  Hence in the limiting scaled process $\bar{s}_{\infty}$,
with the initial state being an empty table,
the load will be $\alpha$ at time
$$\int_{z=0}^\alpha \frac{1}{1-z} dz = \ln \frac{1}{1-\alpha}.$$ That
is, we run until the scaled time $\ln \frac{1}{1-\alpha}$, which corresponds to
(in the unscaled process, asymptotically) $n\ln \frac{1}{1-\alpha}$
time steps.  Alternatively, in the limiting scaled process, at time
$t$, the load is $1-e^{-t}$.

We now turn to understanding the level process, assuming that the
state of the table is fixed according to the values $s_i$.  Again we find it useful to
consider the tails.  Thinking of the unscaled process, so $t$ again
refers to discrete time steps, let $p_i(t)$ be the probability that
the age of the key being placed is at least $i$.  Hence $p_1(t) = 1$
for all time steps.  For $i > 1$, the key being placed at time $t+1$
will have age at least $i$ if and only if both the age of the key being placed
at time $t$ has age at least $i-1$, and cell chosen
by the probe sequence at time $t$ has age at least $i-1$.  This is because,
assuming an empty cell is not found, the younger of the keys will be the
key being placed at the next time step.  
In equation form, we simply have
\vspace{-0.1in}
\begin{eqnarray}
p_i(t+1) = p_{i-1}(t) x_{i-1}(t).
\end{eqnarray}
In the scaled time setting, this can be written as
\vspace{-0.05in}
\begin{eqnarray}
\label{pi}
p_i(t+1/n) = p_{i-1}(t) s_{i-1}(t).
\end{eqnarray}
Notice that the simplicity of this equation helps justify our decision
to focus on variables that represent the tails of the loads.  

Since $s_1$ is bounded by the final load $\alpha$, at each step
with probability at least $1-\alpha$ a key is placed and the level
returns to 1.  Hence the Markov process over the levels is ergodic.
Indeed, standard methods show that for any constant $\epsilon$ this
Markov can be made $\epsilon$-close in statistical distance to its
equilibrium distribution after some corresponding constant number of
placement steps.  We re-emphasize the intuition; in the scaled
process, the $s_i$ values change significantly (that is, by
$\Omega(1)$) only after $\Omega(1)$ scaled time steps (or $\Omega(n)$
unscaled time steps), while the Markov chain governing the $p_i$
converges (arbitrarily closely) to its stationary distribution in
$o(n)$ unscaled time steps.  Hence, it make sense in the limit to
treat the $p_i$ values as fixed in equilibrium given the $s_i$
values.  

It follows from Equation~(\ref{pi}) that the equilibrium for the level process 
satisfies
\vspace{-0.05in}
\begin{eqnarray}
\label{pie}
p_i = p_{i-1} s_{i-1}
\end{eqnarray}
when we treat the $s_i$ as fixed and we use $p_i$ without the $t$ to denote
the equilibrium for the $p_i$ values given the $s_i$ values at that time.

With this we turn our attention to the limiting equations for the $s_i$ in
$\bar{s}_{\infty}$.  
Note that $s_1$ increases whenever a empty cell is hit.  Hence
\vspace{-0.05in}
\begin{eqnarray}
\label{s1}
\frac{ds_1}{dt} = 1-s_1.  
\end{eqnarray}
Integrating, and using $s_1(0) = 0$, this gives
$s_1(t) = 1-e^{-t}$, matching our previous warm-up analysis.
For $s_i$ when $i > 1$, Equation~(\ref{s1}) generalizes to
\vspace{-0.05in}
\begin{eqnarray}
\label{si}
\frac{ds_i}{dt} = p_i(1-s_i),  
\end{eqnarray}
since a cell containing a key with age at least $i$ is created whenever
the age of the key being placed is at least $i$ and the probe sequence
finds either an empty cell or a cell containing a key with age less than $i$. 

At any time $t$, let $\beta(t)$ be the corresponding load at that time.
(Recall we use $\alpha$ for the ``final'' load.)  We use $\beta$ for $\beta(t)$
where the meaning is clear.
Since $\beta(t) = 1 - e^{-t}$, we have
using 
$$\frac{d\beta}{dt} = e^{-t} = 1 - \beta.$$
At the possible risk of confusion, but to avoid conversions back and
forth, we use $s'_i(\beta)$ to represent $s_i$ taken as a function of
the load $\beta$ instead of as a function of time.   We have from the above
that in the setting of the asymptotic limit $s'_i(\beta) = s_i(-\ln(1-\beta))$.
With the expressions for $\frac{d\beta}{dt}$ and $\frac{ds_i}{dt}$ we obtain
the following form for 
$s'_i(\beta)$ as a function of $\beta$ for $i \geq 1$.  
\vspace{-0.05in}
\begin{eqnarray}
\label{si3}
\frac{ds'_i}{d\beta} = \frac{p_i (1-s'_i)}{1-\beta}.  
\end{eqnarray}

Given our equations for $p_i$, we can substitute so that all equations
are in terms of the $s_i$.  Specifically, the equation for
$\frac{ds_i}{dt}$ (or $\frac{ds_i}{d\beta}$) depends only on values
$s_j$ with $j \leq i$.  That is, 
\vspace{-0.05in}
\begin{eqnarray}
\frac{ds_i}{dt} & = & p_i(1-s_i) \\
                & = & p_{i-i}s_{i-1}(1-s_i) \\
                & = & (1-s_i) \prod_{j=1}^{i-1}s_{j}.
\end{eqnarray}
The differential equations can therefore be
solved numerically for $s_i$ values up to any desired constant $K$.  Moreover,
we can truncate the infinite system of differential equations to a
finite system by considering the equations for $\frac{ds_i}{dt}$ up to the
constant $K$.  Because of this, using the large deviation theory, we may formally state
the following:
\begin{theorem} \label{thm:inapp} For any fixed constant $K$ and any $\alpha < 1$, 
for  $i \leq K$ let $s'_i(\alpha)$ 
be the solution for the $s'_i$ at final load $\alpha$ from  the family
of differential equations given by Equation (\ref{si3}) above.  For  $1 \leq i \leq K$, let $X_{i,n}$ be the random variable denoting the fraction of cells with keys of
age at least $i$ using Robin Hood hashing at final load $\alpha$ with $n$ cells.
Then for any $\epsilon > 0$, for sufficiently large $n$
\begin{eqnarray}
\Pr\left( \left|X_{i,n}-s'_i(\alpha)\right|>\epsilon \right)
            & \leq & C_1 \exp(-nC_2(\epsilon)),
\end{eqnarray}
where $C_1$ is a constant that depends on $K$ and $\alpha$ and $C_2(\epsilon)$ is a
constant that depends on $K$, $\alpha$, and $\epsilon$.
\end{theorem}
\begin{proof}
The result follows from Theorem~\ref{ldthm}.  While Theorem~\ref{ldthm}
is stated in terms of time instead of load, this difference in not consequential;
the straightforward proof is given in an appendix.
\end{proof}

\vspace{-0.15in}
\subsection{Implications for the Age Distribution}
\vspace{-0.05in}
\label{sec:agedist}

In \cite{celisphd}, Celis derives the age distribution under Robin Hood
hashing by providing a recurrence.  We demonstrate that this
result also follows from our differential equations analysis.
We note that we use a different notation; the following theorem
corresponds to Theorem {3.1} of \cite{celisphd}.
\begin{theorem} 
\label{celistheorem}
In the asymptotic model for an infinite Robin Hood hash table
with load factor $\beta$ ($\beta < 1$), the fraction $s'_i(\beta)$
of cells containing keys of age at least $i$ is given by
\begin{eqnarray}
\label{srec}
s'_{i+1}(\beta) = 1 - (1 - \beta) e^{\sum_{j=1}^i s'_j(\beta)}.
\end{eqnarray}
\end{theorem} 
\begin{proof}
As standard techniques can be used to show that our family of differential equations
has a unique solution, we show that the recurrence of Equation~(\ref{srec}) satisfies Equation~(\ref{si3}).
We first note the following useful fact:
$$\sum_{j=1}^i \frac{ds'_j}{d\beta} = \frac{1-\prod_{j=1}^i s'_j}{1-\beta}.$$
This follows easily from Equation~(\ref{si3}) by induction, using that $p_i = \prod_{j=1}^{i-1} s_j = \prod_{j=1}^{i-1} s'_j$,
as is easily derived from Equation~(\ref{pie}).  

Now taking the derivative of Equation~(\ref{srec}) we find
\begin{eqnarray*}
\frac{ds'_{i+1}(\beta)}{d\beta} & = & e^{\sum_{j=1}^i s'_j} - (1 - \beta) \left ( \sum_{j=1}^i \frac{ds'_j}{d\beta} \right ) e^{\sum_{j=1}^i s'_j} \\
& = &  \frac{1-s'_{i+1}}{1-\beta} - (1-s'_{i+1})\frac{1-\prod_{j=1}^i s'_j}{1-\beta} \\ 
& = &  \frac{ \left ( {\prod_{j=1}^i} s'_j \right ) (1-s'_{i+1})}{1-\beta} \\
& = &  \frac{p_{i+1}(1-s'_{i+1})}{1-\beta}.
\end{eqnarray*}
Hence the recurrence of Equation~(\ref{srec}) satisfies Equation~(\ref{si3}) as claimed.
\end{proof}

\vspace{-0.15in}
\subsection{Implications for Maximum Age}
\vspace{-0.05in}

We now show, following an approach established in \cite{MitzThesis},
that the fact that the growth of maximum age grows double
logarithmically in $n$ appears as a natural consequence of the differential
equations.  As noted, the general large deviation results we apply
only hold for finite-dimensional systems, so their application can only
apply up to constant ages.  (Of course, one can choose a very large constant
age, so our extension here is clearly primarily of theoretical interest.)
Explicitly proving an $O(\log \log n)$ bound on the maximum can be accomplished
by translating the differential equation argument to a layered induction
argument, successively bounding the fraction of cells holding keys of
age $i$ for each $i$.  While it does not appear motivated by the
differential equations approach, a previous argument in roughly the layered induction
style appears already in \cite{devroye2004worst}, formally providing a
$\log \log n +O(1)$ bound.  Their analysis is very different, however, as it
does not make use of the underlying Markov chain directly, but bounds the behavior
of what it calls the ``head'' and the ``belly'' of the process over stages to resolve
the collision that arises.  Our goal in this analysis is to first show how the
fluid limit analysis provides novel insight into how the doubly exponential decrease
in the age distribution arises.  We then use the fluid limit argument as a guide, leading 
to an alternative (and we believe somewhat simpler) layered induction proof
for a $\log \log n +O(1)$ bound on the maximum age.

\begin{theorem} 
In the asymptotic model for an infinite Robin Hood hash table
with load factor $\alpha < 1$, for sufficiently
large constants $i$, the fraction $s'_i(\alpha)$
of cells that contain keys of age at least $i$ satisfies
\begin{eqnarray}
s'_{i}(\alpha) \leq c_1 c_2^{2^{i-c_3}}
\end{eqnarray}
for some constants $c_1,c_3 > 0$ and $c_2 < 1$ that may depend on $\alpha$.
\end{theorem} 
\begin{proof}
In what follows here let $u = -\ln(1-\alpha)$.  
Let $j$ be the smallest value such that $s_j(u) < u^{-1}$.
Let $s_j(u) = \nu$ and $\nu \cdot u = \nu^* < 1$.  
We remark that while it may not be immediately clear that the $s_j$ go to 0,
it is shown in \cite[Lemma 4]{devroye2004worst} that the $s_i$ have geometrically decreasing tails,
so that $j$ is in fact a constant.  (Alternatively, since the $p_j$ clearly have geometrically
decreasing tails, it follows readily that the $s_j$ 
do as well.)
Now below we use the differential equations
based on time, up until time $u$, so $t \leq u$.    
As $p_{j+1} \leq s_j$, and $s_j$ is increasing over time,
$$\frac{ds_{j+1}(t)}{dt} = p_{j+1}(1-s_{j+1}) \leq s_j \leq \nu.$$  
To reach load $\alpha$ we run for time $u = \ln \frac{1}{1-\alpha}$, and hence 
$$s'_{j+1}(\alpha) =  s_{j+1}(u)
\leq \nu \cdot u = \nu^*.$$   

Inductively, we now find by the same argument that for $k \geq 1$,
$$p_{j+k} \leq (\nu^*)^{2^{k-1}-1}\nu \ \ ; \ \ s_{j+k} \leq (\nu^*)^{2^{k-1}},$$
and the result follows.  That is, 
$$p_{j+k} = p_j \prod_{\ell=0}^{k-1} s_{j+\ell} \leq \nu \prod_{\ell=1}^{k-1} s_{j+\ell} \leq (\nu^*)^{2^{k-1}-1}\nu,$$
where the last step follows from the inductive hypothesis.  Further, as
$$\frac{ds_{j+k}(t)}{dt} \leq p_{j+k} \leq (\nu^*)^{2^{k-1}-1}\nu,$$ we
have 
$$s'_{j+k}(\alpha) \leq (\nu^*)^{2^{k-1}-1}\nu u = (\nu^*)^{2^{k-1}}.$$
The theorem follows.
\end{proof}

As one might hope from previous work (e.g., \cite{MitzThesis}), this fluid
limit argument can be transformed into a layered induction argument to prove
a $\log \log n +O(1)$ upper bound on the maximum age, as we now show.   

\begin{theorem}
Let $M_n$ be the maximum age in a Robin Hood hash table with $n$ cells  and load
factor $\alpha < 1$.  There is a constant $C$ depending only on $\alpha$
such that 
$$\lim_{n \rightarrow \infty} \Pr(M_n \geq \log \log n +C) = 0.$$
\end{theorem}
\begin{proof}  
Instead of thinking about the Robin Hood hash table, we work
with the Markov process we have been considering, where
here we take the time $t$ to be the number of unscaled time steps,
$S_i(t)$ to be the number of cells holding a key of age at least $i$ at
time $t$ (where we include the key being placed in the count if its age
is at least $i$), $s_i(t)$ to be $S_i(t)/n$, and $P(t)$ to be the age of the
key being placed at the $t$th time step.  The idea of the proof is to
start with a bound on $S_j(t)$ at the end of the process for some
useful starting value $j$, and use this to bound the number of steps
for which $P(t) \geq j+1$ over the course of the process.  With this, we in turn bound the number of
steps that yield a new key of age at least $j+1$ to bound $S_{j+1}(t)$
at the end of the process, providing our induction.  This induction,
which uses Chernoff bounds, breaks down at some age (where the Chernoff
bound will no longer readily apply), at which point
we explicitly have to carefully cap off the induction.  We follow
the intuition from our argument for the limiting system.  Finally,
we note for this proof we have not attempted to optimize the constant $C$.

Here we let $u = -\ln(1-(\alpha+\epsilon_0))$ for a small constant
$\epsilon_0 < (1-\alpha)/2$.  We run the process for up to
$nu$ steps; however, once $\alpha n$ keys are placed, we allow the
process to stop, and no action is taken in the remaining time steps.
(Recall that $\alpha n$ keys should be placed after  $-\ln(1-(\alpha))n +o(n)$
steps in expectation.)
Let ${\cal G}$ be the event that $\alpha n$ keys have been placed
after $nu$ steps.  Using standard martingale arguments, one can easily
show $\Pr({\cal G}) = 1-o(1)$; we return to this point at the end of
the analysis.

For convenience we take the case where $u > 1/2$, so $j \geq 2$; this
suffices as smaller $\alpha$ and hence values of $u$ will have
strictly smaller load.
After $nu$ steps, the total age of all keys in the table is at most
$nu$.  Therefore the fraction of keys in the table with age at least
$j=\lceil 16u^3 \rceil$ is at most $(4u)^{-2})$,  
so we deterministically have $s_j(nu) \leq (4u)^{-2}$.  As a warm up, following the 
fluid limit argument, let $X^{j+1}_t$ be 1 if $P(t) \geq j+1$
and 0 otherwise.  For $P(t) \geq j+1$, the cell chosen at the
previous step must have age at least $j$.  For each such time step,
the probability such a cell was chosen is dominated by an independent
Bernoulli trial of probability $s_j(nu) \leq (4u)^{-2}$.  Hence
if we let $X_{j+1} = \sum_{t=1}^{nu} X^{j+1}_t$, we have 
$$\Pr(X_{j+1} \geq (2nu)(4u)^{-2}) \leq \Pr(B(nu,(4u)^{-2}) \geq (2nu)(4u)^{-2}) \leq
\e^{-(nu)(4u)^{-2}/3}.$$
The last part of this equation follows from a standard Chernoff bound
(e.g., \cite{MU}[Theorem 4.4]).  

In turn, $S_{j+1}(t)$ can only increase when $P(t) \geq j+1$, and hence, 
with high probability $S_{j+1}(nu) \leq (2nu)(4u)^{-2}$, so 
$s_{j+1}(nu) \leq (4u)^{-1}/2$.  (This bound is weaker than just taking 
$s_{j+1}(nu) \leq s_{j}(nu)(4u)^{-2}$, but will suffice for our induction.)

Let $X_{j+k}$ similarly be the number of steps where $P(t) \geq
j+k$.  For $k\geq 1$, let $\gamma_k = (4u)^{-2^{k-1}}$;
let $\gamma_0 =  (4u)^{-2}$.  We
inductively show that with high probability
$s_{j+k}(nu) \leq \gamma_k$ by showing $X_{j+k} \leq n \gamma_k$,
as $X_{j+k}$ clearly bounds the number of keys with age $j+k$
placed in the table.  To take care of the conditioning issues
arising in the induction, we establish a simpler but useful
dominating process.  Consider a process $C$ with states $j, \ldots, j+k$
that initially starts at state $j$ and transitions from state
$a$ to $b$ with probability $q(a,b)$ at each time step as follows:
\begin{eqnarray*}
q(a,b) & = & \left\{ \begin{array}{ll}
              s_a(t) & \mbox{for } b = a+1; \\
              s_{j+k}(t) & \mbox{for } a= b= j+k; \\
              s_b(t) -s_{b+1}(t) & \mbox{for } j < b \leq a \leq j+k \mbox{ (except as above)}; \\
              1 - s_{j+1}(t) & \mbox{for } b=j;
              \end{array} \right.
\end{eqnarray*}
That is, our process $C$ behaves like the age of the key
being placed in our original chain, except that the state simply
returns to state $j$ as a default (instead of some smaller age)
and the maximum state we are concerned with is $j+k$.
Because of this, the number of steps where $C$ 
reaches state $j+k$ over $nu$ times steps stochastically dominates
the number of times the original process has $P(t) \geq j+k$.
Also, in order to have 
for $P(t+k) \geq j+k$ in our original process, the key being placed at time $t$ must have
reached age $j$, and then, over the next $k$ steps, the ages of the cells
chosen must be at least $j,j+1,\ldots,j+k-1$.
Similarly, for process $C$ to reach state $j+k$ the states over the last $k$ steps must 
successively increment by 1 through $j,j+1,\ldots,j+k-1$.  

Let $Y_{j+k}$ be the number of steps where $C$ is in state $j+k$
when starting from state $j$ over $nu$ steps.  
By the dominance of chain $C$, 
\begin{eqnarray*}
\Pr(X_{j+k} \geq \gamma_k) & \leq & \Pr(Y_{j+k} \geq \gamma_k).
\end{eqnarray*}
In bounding the quantity on the right hand side, we will want to 
focus on the high probability the case that the $s_{j+i}(t)$ values
are well-behaved.  In particular, 
let ${\cal E}_k$ be the event that $s_{j+i}(t) \leq \gamma_i$ for $1 \leq i \leq k$
and all $t \leq nu$.
Note above we have shown that ${\cal E}_1$ holds with probability $1/n^2$;
as part of our proof we show inductively that ${\cal E}_i$ holds with probability $i/n^2$ for $i$
up to a value $k^*$ where $\gamma_{k^*} \leq n^{-2/5}$.  
Our argument will show that $k^* = \log \log n +O(1)$.
At that point we switch to an explicit argument to finish the analysis.

To bound $Y_{j+k}$, we must consider yet another process $C'$, which dominates
$C$ conditioned on ${\cal E}_{k-1}$.  Specifically, for process $C'$, 
again the are states $j, \ldots, j+k$, but the transition from
$a$ to $b$ occurs with probability $q'(a,b)$ at each time step as follows:
\begin{eqnarray*}
q'(a,b) & = & \left\{ \begin{array}{ll}
              \gamma_{a-j} & \mbox{for } b = a+1; \\
              \gamma_{k} & \mbox{for } a= b= j+k; \\
              \gamma_{b-j} & \mbox{for } j < b \leq a \leq j+k \mbox{ (except as above)}; \\
              \mbox{all remaining probability} & \mbox{for } b=j;
              \end{array} \right.
\end{eqnarray*}
Note $q'(a,b)$ defines a proper probability distribution.
Let $Z_{j+k}$ be the number of steps where $C'$ is in state $j+k$
when starting from state $j$ over $nu$ steps.  Then, conditioned on   
${\cal E}_{k-1}$, $Z_{j+k}$ stochastically dominates $Y_{j+k}$ by
a simple coupling, since $q'(a,b) \geq q(a,b)$ whenever $b > j$.  

The benefit working with $C'$ is that the behavior is fixed over time,
while $C$ depends on the $s_j(t)$ values.  This helps us as follows.
We now bound $Z_{j+k}$.  Because of the dependence in the state our process, we must use a
more powerful tail bound; here we use a martingale approach.  
Let $A_i$ be the probability that
be a binary random variable that is $1$ if and only if the process $C'$
is in state $j+k$ at the $i$th step, and 0 otherwise.  We can use a
standard Doob martingale (see e.g. \cite{MU}[Chapter 12.1), where $Z_{j+k} = \sum_{i=1}^{nu}
A_i$.  Note that the $A_i$ are not independent; however, the outcome of each time
step of the process affects at most $2k$ values of $A_i$, and hence
the outcome at each time step changes $A$ by at most $2k$.
Hence, using the standard Azuma-Hoeffding inequality (see e.g. \cite{MU}[Theorem 12.4]), we have
\begin{eqnarray*}
\Pr(Z_{j+k} \geq 2 \E[Z_{j+k}] & \leq & \Pr(|Z_{j+k} - \E[Z_{j+k}]| \geq \E[Z_{j+k}]) \\
& \leq &  2\mbox{exp}(-\E[Z_{j+k}]^2/(8nk^2))
\end{eqnarray*}
Again, to reach state $j+k$ the states over the last $k$ steps must 
successively increment by 1 through $j,j+1,\ldots,j+k-1$.  
Hence
\begin{eqnarray*}
\E[Z_{j+k}] & \leq & (nu) \prod_{\ell=0}^{k-1} \gamma_{j+\ell} \\
                       & = & \frac{n}{4} (4u)^{-2^{k-1}}.
\end{eqnarray*}
Now as stated before let $k^*$ be the smallest value so that 
$\gamma_{k^*} \leq n^{-2/5}$.  Note from the definition of $\gamma_k$ that
$k^* = \log \log n +O(1)$.
It follows that, as long as $k < k^* < \log n$, so that $\E[Z_{j+k}] \geq n^{3/5}$,
\begin{eqnarray*}
\Pr(Z_{j+k} \geq n\gamma_{k}) \leq \Pr(Z_{j+k} \geq 2 \E[Z_{j+k}]) \\
                                \leq 2\e^{-n^{1/5}/(8\log^2 n)} << 1/n^2.
\end{eqnarray*}
In particular, note that, inductively,
\begin{eqnarray*}
\Pr(Y_{j+k} \geq \gamma_k) & \leq & \Pr(\neg{\cal E}_{k-1}) + \Pr(Y_{j+k} \geq \gamma_k~|~{\cal E}_{k-1}) \\
                           & \leq & \frac{k-1}{n^2} + \Pr(Z_{j+k} \geq \gamma_k) \\
                           & \leq & \frac{k-1}{n^2} + \frac{1}{n^2} = \frac{k}{n^2}.
\end{eqnarray*}
Also
\begin{eqnarray*}
\Pr(\neg{\cal E}_{k}) & \leq & \Pr(\neg{\cal E}_{k-1}) + \Pr(\neg{\cal E}_{k}~|~{\cal E}_{k-1}) \\
                      & \leq & \frac{k-1}{n^2} + \frac{1}{n^2} = \frac{k}{n^2}.
\end{eqnarray*}
Now for $k^*$, we can follow the same argument to show that 
\begin{eqnarray*}
\Pr(Z_{j+k^*} \geq 2 n^{3/5} & \leq & \Pr(|Z_{j+k} - \E[Z_{j+k}]| \geq n^{3/5}) \\
& \leq &  2\e^{-n^{1/5}/(8\log^2 n)} << 1/n^2.
\end{eqnarray*}
Hence, if we let  ${\cal F}$ be the event that $s_{j+k^*}(t) \leq 2n^{-2/5}$,
we find again
\begin{eqnarray*}
\Pr(\neg{\cal F}) & \leq & \Pr(\neg{\cal E}_{k^*-1}) + \Pr(\neg{\cal F}~|~{\cal E}_{k-1}) \\
                      & \leq & \frac{k^*-1}{n^2} + \frac{1}{n^2} = \frac{k^*}{n^2}.
\end{eqnarray*}
Now, finally, we note that for any key to have age $j+k^*+3$, the
Robin Hood hashing scheme must choose 3 cells with load at least $j+k^*$ over sequential steps.
Conditioned on ${\cal F}$, the probability that this event, which we denote by ${\cal H}$ occurs is at most $(2n^{-2/5})^3 \cdot (un) = o(1)$.
We conclude
\begin{eqnarray*}
\Pr(M_n \geq j+k^*+3)& \leq & \Pr(M_n \geq j+k^*+3 \mbox{ and } {\cal G}) + \Pr(M_n \geq j+k^*+3 \mbox{ and }\neg{\cal G}) \\
                       & \leq & \Pr({\cal H}) + \Pr(\neg{\cal G}) \\
                       & \leq & \Pr({\cal H}~|~{\cal F}) + \Pr(\neg{\cal F}) + \Pr(\neg{\cal G}) \\
                       & \leq & (2n^{-2/5})^3 \cdot (un) + \frac{k^*}{n^2} + o(1) = o(1).
\end{eqnarray*}
Here $j+k^*+3$ is $\log \log n +C$ for a constant $C$, so the theorem is proven.
\end{proof}

\vspace{-0.15in}
\subsection{Implications for Unsuccessful Search Times}
\vspace{-0.05in}

For insertions-only tables, there is a simple optimization that speeds
up unsuccessful searches over standard search for Robin Hood hashing.  If the key being searched for is at the
$i$th position in its probe sequence, and the probe yields a cell with
a key with age strictly less than $i$, then the key cannot be in the
table.  This is because if the key were in the table it would have
replaced the younger key in this cell.  Making use of this fact 
allows us to short-circuit an unsuccessful search early, before reaching 
an empty cell.  With this optimization, the probability that an unsuccessful search takes
at least $j$ probes (up to the longest probe sequence in the system)
is $\prod_{k=1}^{j-1} s_k$, since on the $k$th probe it would need to find
a cell with age at least $k$ to continue.  In the asymptotic limit, we have
that the expected number of probes for an unsuccessful search is thus
$$\sum_{j=1}^\infty \prod_{k=1}^{j-1} s_k = \sum_{j=1}^\infty j p_j.$$
That is, the expected number of probes for an unsuccessful search is just
the expected age according to the equilibrium distribution given by the $p_j$ 
at the final load $\alpha$.

This brings up a point that we have not mentioned previously.  For
Robin Hood hashing, it is useful to keep track of the age of each key
with each key.  This may require a small number of extra bits per
cell, which is not unreasonable if keys are large.  As we will see,
even for high loads, 3 bits may be sufficient, and 4 bits handles most
insertion-only cases in practice.  Of course one can always re-derive
the age of a key on the fly by re-computing hash values from the key,
but we expect keeping the age of the key would prove more
efficient.  

\vspace{-0.15in}
\section{Handling Deletions}  
\vspace{-0.05in}

In this section, we demonstrate that we can analyze a standard model for deletions whereby we
load a number of keys into the system, and then alternate between
deleting a key chosen uniformly at random from the table, and inserting
new fresh key.  Our goal is to use the differential equations
analysis to consider the long-term behavior and the steady state (if
one exists) of such a system.  We show that we can analyze the
deletion method described in Celis's thesis \cite{celisphd}, which
is the only deletion scheme we know of previously proposed for Robin Hood hashing.
We then introduce and analyze another deletion scheme that we believe could be more 
suitable in practice.  Further discussion comparing the two
approaches appears with simulations appears in the appendix.
Our scheme is much simpler, and interestingly, the equations we
derive to model their behavior are simpler as well.  


\vspace{-0.15in}
\subsection{Deletions with Tombstones}
\vspace{-0.05in}

Celis suggests a scheme for deletions based on a standard approach of
using tombstone entries.  Cells with deleted keys are marked; such
marked entries are called tombstones.  For insertion and search
purposes deleted entries are treated the same as non-deleted entries,
except that when a deleted key is replaced by younger key on an
insertion, the deleted key can be discarded and does not need to be
put back in the table.  As an optimization, if the age of a key being
inserted is equal to the age of a key in a deleted entry, it can
replace the deleted entry.  This maintains the property that if a key
is placed in the cell given by the $i$th entry in its probe sequence,
then for $1 \leq j < i$ there is a key (which may be a deleted key in
a tombstone cell) with age
at most $j$ in the cell corresponding to the $j$th entry of the key's
probe sequence.

This scheme has the obvious problem that the ages of keys in the table
can only become larger.  Hence, there is no ``steady state'' 
for the age of keys.  However, we can still model its behavior, as we
show here.  

We model this deletion strategy in the setting where we first load the
system to a load $\alpha$, and then alternate between deleting a
random key and inserting a new key.  Entirely similar analysis can be
used for the setting where elements in the system all have an
exponentially distributed lifetime with a fixed mean and arrivals of
new keys form a Poisson process of a given rate, in which case
$\alpha$ represents the equilibrium load of the system.  This is a
natural model for studying deletions; for example, it was used in
Celis's thesis, as well as other works.  While this model does not
capture possible worst-case behaviors, and in particular does not
handle re-insertion of deleted keys, it remains a useful model for 
understanding basic performance.  Of course, it also fits well with our
analysis:  we start by running the original set of differential equations
to load $\alpha$, and then from the resulting state we run a new
set of differential equations that take deletions into account.  

We first describe changes to how we think about the state space.
Our level process will require an additional state, call it
state 0, which corresponds to the state when we perform a deletion. 
Also, while we use $s_i$ to again represent the fraction of cells
in the table that have an {\em undeleted} key of age at least $i$, we also
need to track the fraction of cells in the table that are tombstones 
containing {\em deleted} keys of age at least $i$.  Let us refer to 
these as $u_i$.  

We first consider the level process.  
Let us think in unscaled time steps.  Let $q(t)$ be the probability
of being in state 0 -- that is, that we are about to perform a deletion
-- and as before for $i \geq 1$ let $p_i(t)$ be the probability that we are trying
to place a key of age at least $i$ in the table.
The state of the level process is a Markov chain
assuming that the state of the table is fixed at certain values $s_i$ and $u_i$.

The equations for $p_i(t)$ and $q(t)$ are as follows:
\begin{eqnarray}
p_1(t+1) & = & 1- \sum_{j=1}^\infty (p_j(t) - p_{j+1}(t)) (1-s_1 -u_{j+1}); \\
p_i(t+1) & = & p_{i-1}(t) s_{i-1} + \sum_{j=i-1}^\infty (p_j(t) - p_{j+1}(t)) u_{j+1} \ , \ i \geq 2;\\
q(t) & = & \sum_{j=1}^\infty (p_j(t) - p_{j+1}(t)) (1-s_1 -u_{j+1}).
\end{eqnarray}
Our equation for $p_i$ now has an additional term that takes into
account that a key will be placed in a cell with a deleted key when
the deleted key's age is less than or equal to the age of key being
placed.  Similarly, the equation for $q_j$ is based on summing over
each possible age $j$ of the key being placed the probability that it is
placed successfully.  This gives the equilibrium equations:
\begin{eqnarray}
p_1 & = & 1- \sum_{j=1}^\infty (p_j - p_{j+1}) (1-s_1 -u_{j+1}) ;\\
p_i & = & p_{i-1} s_{i-1} + \sum_{j=i-1}^\infty (p_j - p_{j+1}) u_{j+1}  \ , \ i \geq 2 ;\\
q & = & \sum_{j=1}^\infty (p_j - p_{j+1}) (1-s_1 -u_{j+1}).
\end{eqnarray}
We now turn to equations for $s_i$ and $u_i$.  We find
\begin{eqnarray}
\label{sicelis}
\frac{ds_i}{dt} & = & \sum_{j=i}^\infty (p_j - p_{j+1}) (1-s_i -u_{j+1}) - \frac{qs_i}{s_1} ;\\   
\label{uicelis}
\frac{du_i}{dt} & = & \frac{qs_i}{s_1} - \sum_{j=i}^\infty p_j (u_j - u_{j+1}).
\end{eqnarray}
Note that, in the above, the values $s_1$ in the denominator of the terms
of the form $\frac{qs_i}{s_1}$ can be replaced by $\alpha$, since $s_1$ is $\alpha$
once we start performing deletions under our model.

While this system appears inherently infinite-dimensional, it proves
very accurate as shown in simulations.  Again, due to lack of space, further discussion is given
in the full version of the paper.

\vspace{-0.15in}
\subsection{Deletions without Tombstones}
\vspace{-0.05in}

We now suggest and analyze a simpler scheme that we believe remains
effective based on our analysis and simulations.  Deleted entries are
simply deleted from the table; no tombstones are used.  A problem with
this approach is that one can no longer use an empty cell as a
stopping criterion for an unsuccessful search.  Similarly, we no
longer have the property that on a search an occupied cell must have a
key of age at most $i$ on the $i$th probe, or we can declare the
search unsuccessful.  The only way to cope with unsuccessful searches
is to keep track of largest age of any key currently in the system.
This can be done by having the table keep counters of the number of
keys of each age.

As before, let us first consider the level process.  As a reminder
we work in the model where we load the system to a load $\alpha$, and
then alternate between deleting a random key and inserting a new key.  
Our level process will therefore require an additional state, call it
state 0, which corresponds to a deletion.  Let us think in unscaled time steps.  Let $q(t)$ be the probability
of being in state 0 -- that is, that we are about to perform a deletion
-- and as before for $i \geq 1$ let $p_i(t)$ be the probability that we are trying
to place a key of age at least $i$ in the table.
The state is again a Markov chain,
assuming that the state of the table is fixed at certain values $s_i$.

First, let us consider $q(t)$.  When we are placing an item, we complete 
the placement with probability $1-s_1$, the probability of finding a cell
without a key (either from deletion or from being empty).  Note that $1-s_1$
is $1-\alpha$ based on our model.  Hence
$$q(t+1) = p_1(t) (1-\alpha).$$
However, $p_1(t) = 1 - q(t)$.  Substituting gives
$$q(t+1) = (1 - q(t)) (1-\alpha).$$
Letting $q$ be the equilibrium probability for the chain for $q(t)$, we have
$$q = \frac{1-\alpha}{2-\alpha}.$$
Note that this gives the equilibrium probability
$$p_1 = \frac{1}{2-\alpha}.$$
Finally, as with the original Robin Hood process, we have for $i \geq 2$ that
\begin{eqnarray}
\label{pix}
p_i = p_{i-1} s_{i-1}.
\end{eqnarray}

With this we turn our attention to the limiting equations for the $s_i$.  
Note that $s_1$ increases whenever an available cell is found for a placement,
and $s_1$ decreases whenever a deletion occurs.  Hence
\begin{eqnarray}
\label{s1d}
\frac{ds_1}{dt} = p_1(1-s_1) - q.  
\end{eqnarray}
Note that we have $\frac{ds_1}{dt} = 0$ when 
$s_1 = \alpha$, so this equation is consistent with our model.  

For $s_i$ when $i > 1$, Equation~(\ref{s1d}) generalizes to
\begin{eqnarray}
\label{s2d}
\frac{ds_i}{dt} = p_i(1-s_1) - q(s_i/s_1)
\end{eqnarray}
That is, a cell containing a key with age at least $i$ is created whenever
the age of the key being placed is at least $i$ and an empty cell is found by the probe sequence.
A cell containing a key with age at least $i$ is removed whenever a deletion
occurs, and that deletion is for a cell holding a key of age at least $i$,
which occurs with probability $s_i/s_1$.  Assuming we start with alternating deletions and insertions
when $s_1 = \alpha$, 
we can write for $i \geq 1$:
\begin{eqnarray}
\label{s3d}
\frac{ds_i}{dt} = p_i(1-\alpha) - q(s_i/\alpha)
\end{eqnarray}
Note that, under this model, we again have that the $s_i$ term depends only on 
values of $p_j$ and $s_j$ with $j \leq i$;  hence we can apply Theorem~\ref{ldthm} 
over finite time intervals to the truncated family of equations up to $i \leq L$
for any constant $L$ to obtain accurate values for the limiting system.  

For this model, we find that there is a unique equilibrium distribution for the underlying
Equations~(\ref{s3d}) and (\ref{pix});  this gives us an idea as to the long term performance
of this approach.  In equilibrium, using $s_1 = \alpha$, we find that $ds_i/dt = 0$ gives the 
following reasonable equation:
\begin{eqnarray}
\label{recur}
s_i & = & \frac{p_i}{p_i + \frac{1-\alpha}{\alpha(2-\alpha)}}.  
\end{eqnarray}
Equation~(\ref{recur}) along with Equation~(\ref{pix}) can be used to show that
the $s_i$ again decrease double exponentially at some point at the equilibrium
given by the family of Equations~(\ref{s3d}).  

\begin{theorem} 
\label{thm:dex2}
In the asymptotic model for an infinite Robin Hood hash table
with load factor $\alpha < 1$ and alternating deletions, for sufficiently
large constants $i$, the value of $s_i$ at the equilibrium point where 
$ds_i/dt = 0$ everywhere satisfies
\begin{eqnarray}
s_{i} \leq c_1 c_2^{2^{i-c_3}}
\end{eqnarray}
for some constants $c_1,c_3 > 0$ and $c_2 < 1$ that may depend on $\alpha$.
\end{theorem} 
\begin{proof}
In what follows let $s_i$ and $p_i$ refer to their values in equilibrium.
Let $z = \frac{1-\alpha}{\alpha(2-\alpha)}$.  From Equation~(\ref{recur}), $s_i < p_i/z$. 
If $z \geq 1$, using this and $p_i = s_{i-1} p_{i-1}$, we can induct to find
$$s_i \leq  \left(\frac{\alpha}{z(2-\alpha)} \right)^{2^{i-2}}.$$

For case $z < 1$, we note the tails of the $p_i$ must decrease
geometrically, since a key is placed in an empty cell 
or a cell with a deleted item with probability at least $1-\alpha$ at each step.
Let $j$ be the smallest value such that $p_j \leq z^2$.  Then inductively we find 
$$s_{j+k} \leq z^{2^{k}}.$$
\end{proof}

\section{Conclusion}
 
We have shown how to use the framework of Markov chains, often also
called the fluid limit analysis or mean-field approach, to analyze
Robin Hood hashing.  In particular, we have shown that for Robin Hood
hashing the analysis naturally requires the use of an additional level
process.  Besides providing a new way of gaining insight
into previous results, we have shown that our methods lead to a simple
recurrence describing the equilibrium behavior of Robin Hood hashing
under a natural deletion model when not using tombstones.  
 
Robin Hood hashing appears to perform essentially the same whether
using probe sequences based on double hashing and random hashing.
Proving this seems a worthwhile open question.  Relatedly, the recent
work of \cite{mitzenmacherdouble} applies fluid limit analysis to show that double
hashing yields the same behavior as fully random hashing under the
``balanced allocation'' paradigm.  Alternatively, one could try to
extend the approach of \cite{lueker1993more} used for standard open addressing hashing.

\bibliographystyle{plain}

\newpage

\section*{Appendices}
\subsection*{Proof of Theorem~\ref{thm:inapp}}

\begin{proof}
The result follows from Theorem~\ref{ldthm}.  While Theorem~\ref{ldthm}
is stated in terms of time instead of load, this difference in not consequential,
as we explain subsequently.
Let $Y_{i,n}$ be the random variable denoting the fraction of cells with keys of
age at least $i$ using Robin Hood hashing after $-n \ln (1-\alpha)$ unscaled time steps
in a hash table with $n$ cells.
Then Theorem~\ref{ldthm} gives us that 
\begin{eqnarray}
\Pr\left( \left|Y_{i,n}-s_i(-\ln(1-\alpha))\right|>\epsilon \right)
            & \leq & C_3 \exp(-nC_4(\epsilon)),
\end{eqnarray}
where $C_3$ is a constant that depends on $K$ and $\alpha$ and $C_4(\epsilon)$ is a
constant that depends on $K$, $\epsilon$, and $\alpha$.
Note that this depends on our restriction of 
the system to be finite dimensional, and the fact that 
the evolution of
the $s_i$ for $i \leq K$ only depends on the values $s_1,s_2,\ldots,s_K$.
(Again, this is another reason to write equations in terms of
the tails of the loads.)  The conditions of Theorem~\ref{ldthm} are easily checked.
In particular, we have noted the Markov process over levels is ergodic for any
given $s_i$ values.  For the second condition, the transition rates
$\nu_i(\bar{x};m)$ are finite size polynomials of load vector $\bar{x}$. 
This implies that $\log \nu_i(\bar{x};m)$ is Lipschitz continuous in coordinates of
$\bar{x}$.  In particular, the bound $x_i \leq 1$ gives an upper bound. 
If $\nu_i(\bar{x};m) = 0$, then it
means the transition is absent and we neglect it. Otherwise $x_i \geq
1/n_0$, giving a lower bound on the $\log \nu_i(\bar{x};m)$.
This completes the check of condition \ref{c2}.

Now note that $s'_i(\alpha) = s_i(-\ln (1-\alpha))$ in the limiting
system.  We find $X_{i,n}$ and $Y_{i,n}$ differ by $o(1)$ terms with high
probability, and in fact
$$\Pr\left( \left|Y_{i,n} - X_{i,n} \right| > \gamma \right ) \leq C_5 \exp(-nC_6(\gamma)),$$
where $C_5$ is a constant that depends on $K$ and $\alpha$ and $C_6(\gamma)$ is a
constant that depends on $K$, $\gamma$, and $\alpha$.
We sketch the reasoning: consider the coupling where we
perform the Robin Hood process for the maximum of $-n \ln(1 - \alpha)$
time steps and the number of time steps to reach load $\alpha$.  The
load after $-n \ln(1 - \alpha)$ time steps will be $\alpha \pm o(1)$
with high probability, by standard martingale arguments; alternatively,
the number of time steps to reach load $\alpha$ is $-n \ln(1 - \alpha)
+ o(n)$ with high probability by standard Chernoff-type bounds, since
the number of time steps to place each key is an independent geometric
random variable with bounded mean.  The theorem statement holds by
summing the probability that 
$\Pr\left( \left|Y_{i,n}-s_i(-\ln(1-\alpha))\right|>\epsilon/2 \right )$
and 
$\Pr\left( \left|Y_{i,n} - X_{i,n} \right| > \epsilon/2 \right )$.  
\end{proof}

\subsection*{Simulations for Insertions Only}

In this section we provide simulation results.  These results serve
the dual purpose of demonstrating the effectiveness of Robin Hood
hashing and verifying our analysis.  We note that simulation results
were also presented in \cite{celisphd}, and one might look there for
further discussion on the effectiveness of Robin Hood hashing.  Here,
the simulation results are presented for completeness, to provide a
high-level verification of the utility of theoretical framework.

Table~\ref{tab:table1} show results with a load $\alpha$ of $0.95$ on the
hash table.  The fraction of keys in the table with a given age (up to 7) are given.  The
results from the differential equations were calculated using the
standard Euler's method with discrete time steps of length $10^{-6}$;
that is, we calculate successive estimates of the variables from the
differential equations using the derivative at the current values and
advance time in steps of $10^{-6}$.  
We also calculate the results from the recurrence of  Theorem~\ref{celistheorem}.
For our simulation results, the
probe sequence for all of the elements were determined using the
pseudo-random generator drand48; all trials were performed in a single
seeded run.  For the simulations, the expression $x \pm y$ refers to
the average $x$ and standard deviation $y$ over 1000 trials.
Unsurprisingly, the differential equations and the results from Theorem~\ref{celistheorem}
agree quite closely, with the discrepancy explained by our calculation method for the
differential equations.  The theoretical results match the simulations very closely.

These results also show the potential effectiveness of Robin Hood
hashing.  Even at a load of $0.95$, the maximum age over these
simulations was 7.  The average number of probes for a successful
search (of a random key in the table) is approximately $3.15$; the
average number of probes for an unsuccessful search is approximately
$3.59$.  As pointed out by Celis \cite{celisphd}, there are further
possible ways to speed up searches by using procedures other than the
standard search procedure.  For example, since most keys have age 3 or
4 under this load, one can start the search with the third and fourth
entries of the probe sequence to reduce the expected number of cells
examined for a successful search.  Note this would be possible
given a hash function in the form $h : K \times \mathbb{N} \rightarrow
[0,n-1]$, so we can examine the $i$th entry in the probe sequence directly.  

\begin{table}
\center
{\footnotesize
\begin{tabular}{|c|c|c|c|c|c|}
\hline
Key & Differential &  Celis &   Sims  &  Sims &  Sims  \\
Age & Equations  &  Theorem  &  $n = 8192$ & $n = 65536$ & $n =524288$ \\ \hline
1 & 0.083458328 & 0.083458403 & 0.083621434 $\pm$ 0.004965813 & 0.083429753 $\pm$ 0.001831579 & 0.083428593 $\pm$ 0.000614445 \\ \hline
2 & 0.188976794 & 0.188976856 & 0.189157158 $\pm$ 0.008679919 & 0.189075507 $\pm$ 0.003140080 &
0.188945002 $\pm$ 0.001092080 \\ \hline
3 & 0.323793458 & 0.323793385 & 0.323707145 $\pm$ 0.008193784 & 0.323798214 $\pm$ 0.003014369 &
0.323775391 $\pm$ 0.001046901 \\ \hline
4 & 0.303363752 & 0.303363594 & 0.302934079 $\pm$ 0.008740126 & 0.303259705 $\pm$ 0.003202237 &
0.303395377 $\pm$ 0.001123541 \\ \hline
5 & 0.095303269 & 0.095303242 & 0.095385891 $\pm$ 0.009795472 & 0.095321415 $\pm$ 0.003651414 &
0.095341351 $\pm$ 0.001226667 \\ \hline
6 & 0.005092100 & 0.005092104 & 0.005182087 $\pm$ 0.001441913 & 0.005103134 $\pm$ 0.000523627 &
0.005101511 $\pm$ 0.000174886 \\ \hline
7 & 0.000012417 & 0.000012417 & 0.000012208 $\pm$ 0.000039393 & 0.000012271 $\pm$ 0.000014530 &
0.000012775 $\pm$ 0.000005037 \\ \hline
\end{tabular}
}
\caption{Results for Robin Hood hashing, both theoretical and from simulations.  The simulations
are for $\alpha = 0.95$ and 1000 trials.}
\label{tab:table1}
\end{table}

\begin{table}
\center
{\footnotesize
\begin{tabular}{|c|c|c|c|c|c|}
\hline
Key & Differential &  Celis &   Sims  &  Sims &  Sims  \\
Age & Equations  &  Theorem  &  $n = 8192$ & $n = 65536$ & $n =524288$ \\ \hline
1 & 0.083458328 & 0.083458403 & 0.083847726 $\pm$ 0.004951031 & 0.083431054 $\pm$ 0.001784843 & 
0.083421016 $\pm$ 0.000612585 \\ \hline
2 & 0.188976794 & 0.188976856 & 0.189757389 $\pm$ 0.008688261 & 0.188963523 $\pm$ 0.003113784 &
0.188964048 $\pm$ 0.001082245 \\ \hline
3 & 0.323793458 & 0.323793385 & 0.324344385 $\pm$ 0.008153945 & 0.323768419 $\pm$ 0.002991134 &
0.323744102 $\pm$ 0.001028609 \\ \hline
4 & 0.303363752 & 0.303363594 & 0.302252763 $\pm$ 0.008952148 & 0.303332755 $\pm$ 0.003156650 &
0.303420142 $\pm$ 0.001083855 \\ \hline
5 & 0.095303269 & 0.095303242 & 0.094714212 $\pm$ 0.009595910 & 0.095374789 $\pm$ 0.003633155 &
0.095344666 $\pm$ 0.001237186 \\ \hline
6 & 0.005092100 & 0.005092104 & 0.005069262 $\pm$ 0.001446734 & 0.005116963 $\pm$ 0.000521156 &
0.005093440 $\pm$ 0.000182278 \\ \hline
7 & 0.000012417 & 0.000012417 & 0.000014264 $\pm$ 0.000041576 & 0.000012496 $\pm$ 0.000014498 &
0.000012587 $\pm$ 0.000005136 \\ \hline
\end{tabular}
}
\caption{Results for Robin Hood hashing, both theoretical and from simulations, but here
the simulations use double hashing.}
\label{tab:table2}
\end{table}

It is worth noting (as was noted in \cite{celisphd} as well) that the results
appear essentially unchanged even if double hashing is used instead of
(our proxy for) fully random hashing.  (In double hashing, we choose a
starting point $a$ and an offset $b$ that is relatively prime to the
table size and our probe sequence is given by $a, a+b, a+2b,\ldots$,
where the values in the probe sequence are taken modulo the size of
the hash table.)  Table~\ref{tab:table2} shows a representative example.  We
are not sure yet how to prove this, although we suspect that theoretical techniques that
have shown double hashing has the same performance as fully random hashing
in other settings may apply (e.g., \cite{lueker1993more,mitzenmacherdouble}). 
The challenge lies in accounting for the ages of placed keys in such an analysis.

\subsection*{Simulations with Deletions}

We first consider the setting with tombstones.  Here our goal in the
simulation is simply to show that the proposed differential equations
accurately model the actual system for finite periods of time
correctly.  In order to keep tables reasonably sized, the result of
Table~\ref{tab:tablewithtomb} shows results with a load $\alpha$ of
$0.9$; here we load the table with $0.9n$ items, and then alternately
delete and insert items until $2n$ items have been placed.  For the
differential equations we again use Euler's method with discrete time
steps of length $10^{-6}$.  Here we show the fraction of keys in the
system (not including tombstones) by age.  The results are very
accurate, and begins to show the effect of using tombstones; the keys
of ages 1 and 2 are vanishing, and the ages of the keys in the cells
are increasing over time.  If we continued the simulation further, we
would see ages continue to grow well beyond 18.

\begin{table}
\center
{\footnotesize
\begin{tabular}{|c|c|c|c|c|}
\hline
Key & Differential &   Sims  &  Sims &  Sims  \\
Age & Equations  &   $n = 8192$ & $n = 65536$ & $n =524288$ \\ \hline
1 & 0.0000000012 & 0 $\pm$ 0 & 0 $\pm$ 0 & 0 $\pm$ 0 \\ \hline
2 & 0.0000000088 & 0 $\pm$ 0 & 0 $\pm$ 0 & 0.0000000042 $\pm$ 0.0000000947 \\ \hline
3 & 0.0000000621 & 0.0000004069 $\pm$ 0.0000074176 & 0.0000000509 $\pm$ 0.0000009272 & 0.0000000721 $\pm$ 0.0000003956 \\ \hline
4 & 0.0000004128 & 0.0000001356 $\pm$ 0.0000042869 & 0.0000003730 $\pm$ 0.0000025999 & 0.0000004472 $\pm$ 0.0000009996 \\ \hline
5 & 0.0000025165 & 0.0000016276 $\pm$ 0.0000147681 & 0.0000024584 $\pm$ 0.0000062977 & 0.0000026046 $\pm$ 0.0000023888 \\ \hline
6 & 0.0000140033 & 0.0000168181 $\pm$ 0.0000478803 & 0.0000138008 $\pm$ 0.0000151588 & 0.0000140084 $\pm$ 0.0000054692 \\ \hline
7 & 0.0000711550 & 0.0000779872 $\pm$ 0.0001004372 & 0.0000708691 $\pm$ 0.0000354163 & 0.0000718202 $\pm$ 0.0000124441 \\ \hline
8 & 0.0003302926 & 0.0003438220 $\pm$ 0.0002233071 & 0.0003195890 $\pm$ 0.0000746674 & 0.0003317559 $\pm$ 0.0000275909 \\ \hline
9 & 0.0014006589 & 0.0014721280 $\pm$ 0.0005133678 & 0.0013790479 $\pm$ 0.0001734967 & 0.0014051570 $\pm$ 0.0000624493 \\ \hline
10 & 0.0054203409 & 0.0056228130 $\pm$ 0.0012861569 & 0.0053175036 $\pm$ 0.0004155336 & 0.0054218866 $\pm$ 0.0001541873 \\ \hline
11 & 0.0190426783 & 0.0195947376 $\pm$ 0.0033890440 & 0.0187390051 $\pm$ 0.0011403881 & 0.0190677512 $\pm$ 0.0004144598 \\ \hline
12 & 0.0596408085 & 0.0610249559 $\pm$ 0.0089433574 & 0.0587034689 $\pm$ 0.0029890910 & 0.0597147622 $\pm$ 0.0010595725 \\ \hline
13 & 0.1576758321 & 0.1599762648 $\pm$ 0.0177522046 & 0.1557674884 $\pm$ 0.0061968220 & 0.1578802884 $\pm$ 0.0021976097 \\ \hline
14 & 0.3056376472 & 0.3063086939 $\pm$ 0.0174270189 & 0.3036759011 $\pm$ 0.0063584378 & 0.3057603924 $\pm$ 0.0022285217 \\ \hline
15 & 0.3239990269 & 0.3198004883 $\pm$ 0.0176422773 & 0.3255669696 $\pm$ 0.0058555114 & 0.3238154957 $\pm$ 0.0021086756 \\ \hline
16 & 0.1187678782 & 0.1173343280 $\pm$ 0.0260781256 & 0.1218899664 $\pm$ 0.0094442642 & 0.1185390890 $\pm$ 0.0033043206 \\ \hline
17 & 0.0079676382 & 0.0083817985 $\pm$ 0.0044648123 & 0.0085186328 $\pm$ 0.0015278714 & 0.0079452633 $\pm$ 0.0005104320 \\ \hline
18 & 0.0000292668 & 0.0000429947 $\pm$ 0.0000969077 & 0.0000348751 $\pm$ 0.0000273052 & 0.0000292015 $\pm$ 0.0000087025 \\ \hline
\end{tabular}
}
\caption{Results for Robin Hood hashing with deletions and tombstones, both theoretical and from simulations.  
Here $2n$ total items are inserted.  The simulations are for $\alpha = 0.90$ and 1000 trials.}
\label{tab:tablewithtomb}
\end{table}

Table~\ref{tab:tablenotomb} shows results with the same setup for the
load and deletion pattern.  Here the maximum age does not increase the
same way, as noted in our analysis, and the maximum age remains
smaller at 12.  Again, the differential equations prove highly
accurate.  In this case, however, we are further interested in the
equilibrium distribution of the age as given by
Equation~(\ref{recur}).  As a proxy we run the simulations again, but
after loading the table with $0.9n$ items, we alternately delete and
insert items until $10n$ items have been placed.  
Table~\ref{tab:tablenotombeq} shows these results, compared to the 
calculated equilibrium distribution from Equation~(\ref{recur}) (which
was derived from the corresponding differential equations).  Again,
we see that the results match well, showing the utility of the differential
equation.  Also, as shown Theorem~\ref{thm:dex2}, we see the probability a
cell has a certain age falls very quickly (doubly exponentially) at the tail
of the distribution.  The average time for a successful search naturally 
converges to 10;  the maximum age, and hence the time for an unsuccessful
search, is only around 16.  

\begin{table}
\center
{\footnotesize
\begin{tabular}{|c|c|c|c|c|}
\hline
Key & Differential &   Sims  &  Sims &  Sims  \\
Age & Equations  &   $n = 8192$ & $n = 65536$ & $n =524288$ \\ \hline
1 & 0.0109912456 & 0.0110180068 $\pm$ 0.0013529380 & 0.0110012567 $\pm$ 0.0005847770 & 0.0109806228 $\pm$ 0.0003822734 \\ \hline
2 & 0.0132627846 & 0.0132962114 $\pm$ 0.0014180089 & 0.0132489158 $\pm$ 0.0006426987 & 0.0132499410 $\pm$ 0.0004533081 \\ \hline
3 & 0.0164099523 & 0.0164749126 $\pm$ 0.0017334129 & 0.0164142764 $\pm$ 0.0007593122 & 0.0163919906 $\pm$ 0.0005562663 \\ \hline
4 & 0.0215495612 & 0.0216120513 $\pm$ 0.0020547157 & 0.0215255058 $\pm$ 0.0009606507 & 0.0215314282 $\pm$ 0.0007204324 \\ \hline
5 & 0.0330311741 & 0.0333027883 $\pm$ 0.0029126093 & 0.0329123421 $\pm$ 0.0014204389 & 0.0330133628 $\pm$ 0.0010944805 \\ \hline
6 & 0.0655968369 & 0.0660263377 $\pm$ 0.0054606729 & 0.0653358713 $\pm$ 0.0026905826 & 0.0655498797 $\pm$ 0.0021627231 \\ \hline
7 & 0.1508513719 & 0.1515972160 $\pm$ 0.0111574845 & 0.1502017361 $\pm$ 0.0059413695 & 0.1507750735 $\pm$ 0.0049173364 \\ \hline
8 & 0.2865694087 & 0.2866796841 $\pm$ 0.0139163312 & 0.2856982460 $\pm$ 0.0097570156 & 0.2864050409 $\pm$ 0.0091539804 \\ \hline
9 & 0.2955178737 & 0.2936264874 $\pm$ 0.0145071670 & 0.2955954899 $\pm$ 0.0101189923 & 0.2951344134 $\pm$ 0.0094283266 \\ \hline
10 & 0.1003906169 & 0.0994804734 $\pm$ 0.0148415190 & 0.1011552450 $\pm$ 0.0060279724 & 0.1001607244 $\pm$ 0.0036143267 \\ \hline
11 & 0.0058132004 & 0.0058682671 $\pm$ 0.0020576713 & 0.0058955320 $\pm$ 0.0007574272 & 0.0057927890 $\pm$ 0.0003043712 \\ \hline
12 & 0.0000159736 & 0.0000185627 $\pm$ 0.0000504047 & 0.0000165817 $\pm$ 0.0000181370 & 0.0000157326 $\pm$ 0.0000057864 \\ \hline
\end{tabular}
}
\caption{Results for Robin Hood hashing with deletions and no tombstones, both theoretical and from simulations.  
Here $2n$ total items are inserted.  The simulations are for $\alpha = 0.90$ and 1000 trials.}
\label{tab:tablenotomb}
\end{table}

\begin{table}
\center
{\footnotesize
\begin{tabular}{|c|c|c|c|c|}
\hline
Key & Calculated &   Sims  &  Sims &  Sims  \\
Age & Equilibrium  &   $n = 8192$ & $n = 65536$ & $n =524288$ \\ \hline
1 & 0.0109890110 & 0.0110295238 $\pm$ 0.0013242152 & 0.0109825070 $\pm$ 0.0005826258 & 0.0109803031 $\pm$ 0.0003835962 \\ \hline
2 & 0.0132380001 & 0.0132215540 $\pm$ 0.0014401913 & 0.0132458671 $\pm$ 0.0006504940 & 0.0132260213 $\pm$ 0.0004546164 \\ \hline
3 & 0.0162108987 & 0.0162228928 $\pm$ 0.0016968246 & 0.0161892632 $\pm$ 0.0007604040 & 0.0161916268 $\pm$ 0.0005516735 \\ \hline
4 & 0.0202345136 & 0.0202062958 $\pm$ 0.0019607504 & 0.0202353850 $\pm$ 0.0009027583 & 0.0202174129 $\pm$ 0.0006767551 \\ \hline
5 & 0.0258283516 & 0.0259177958 $\pm$ 0.0022695870 & 0.0258402714 $\pm$ 0.0010973218 & 0.0258078395 $\pm$ 0.0008603925 \\ \hline
6 & 0.0338433436 & 0.0339240307 $\pm$ 0.0027397431 & 0.0338269272 $\pm$ 0.0013708900 & 0.0338119259 $\pm$ 0.0011128289 \\ \hline
7 & 0.0457090363 & 0.0458778980 $\pm$ 0.0034805126 & 0.0457426674 $\pm$ 0.0017698065 & 0.0456839344 $\pm$ 0.0014981448 \\ \hline
8 & 0.0638449846 & 0.0640809076 $\pm$ 0.0043114573 & 0.0638339146 $\pm$ 0.0024221701 & 0.0637959164 $\pm$ 0.0020769152 \\ \hline
9 & 0.0921369579 & 0.0923216211 $\pm$ 0.0056647485 & 0.0921366789 $\pm$ 0.0034257988 & 0.0920613574 $\pm$ 0.0029824930 \\ \hline
10 & 0.1351848968 & 0.1354930164 $\pm$ 0.0077845001 & 0.1351204660 $\pm$ 0.0048713224 & 0.1350763809 $\pm$ 0.0043471712 \\ \hline
11 & 0.1893101510 & 0.1891220132 $\pm$ 0.0091155314 & 0.1891978863 $\pm$ 0.0064652559 & 0.1891614878 $\pm$ 0.0060426133 \\ \hline
12 & 0.2098741222 & 0.2093249217 $\pm$ 0.0100319761 & 0.2096149248 $\pm$ 0.0071636039 & 0.2097165961 $\pm$ 0.0067030989 \\ \hline
13 & 0.1226847741 & 0.1217393755 $\pm$ 0.0111532706 & 0.1223768263 $\pm$ 0.0054076900 & 0.1225484444 $\pm$ 0.0040878863 \\ \hline
14 & 0.0205100133 & 0.0201279800 $\pm$ 0.0038388175 & 0.0202694291 $\pm$ 0.0015550304 & 0.0203358892 $\pm$ 0.0007973781 \\ \hline
15 & 0.0004008004 & 0.0003907662 $\pm$ 0.0002525555 & 0.0003878323 $\pm$ 0.0000920074 & 0.0003857338 $\pm$ 0.0000357398 \\ \hline
16 & 0.0000001447 & 0.0000004065 $\pm$ 0.0000074139 & 0.0000001524 $\pm$ 0.0000016004 & 0.0000001291 $\pm$ 0.0000005412 \\ \hline
\end{tabular}
}
\caption{Results for Robin Hood hashing with deletions and no tombstones, both theoretical and from simulations.  
Here $10n$ total items are inserted, and compared against the calculated equilibrium distribution.  The simulations are for $\alpha = 0.90$ and 1000 trials.}
\label{tab:tablenotombeq}
\end{table}

\end{document}